\documentclass[prd,onecolumn,notitlepage,amsmath,nofootinbib,superscriptaddress,showpacs]{revtex4-1}

\usepackage[brazil,english]{babel}
\usepackage[utf8]{inputenc}
\usepackage{amsmath}
\usepackage{amssymb}

\usepackage{ulem}
\PassOptionsToPackage{normalem}{ulem}

\usepackage{amsfonts}
\usepackage{graphicx}

\makeatletter

\newtheorem{theorem}{}

\newtheorem{proposition}[theorem]{Proposition}

\newenvironment{proof}[1][Proof]{\noindent\textbf{#1.} }{\ \rule{0.5em}{0.5em}}

\makeatother

\begin{document}

\title{A Hamiltonian approach to Thermodynamics}


\author{M. C. Baldiotti}
\email{baldiotti@uel.br}
\affiliation{Departamento de F\'{\i}sica, Universidade Estadual de Londrina, 86051-990,
Londrina-PR, Brazil.}

\author{R. Fresneda}
\email{rodrigo.fresneda@ufabc.edu.br}
\affiliation{Universidade Federal do ABC, Av. dos Estados 5001, 09210-580, Santo
Andr\'{e}-SP, Brazil}

\author{C. Molina}
\email{cmolina@usp.br}
\affiliation{Escola de Artes, Ci\^{e}ncias e Humanidades, Universidade de S\~{a}o Paulo \\ 
Av. Arlindo Bettio 1000, CEP 03828-000, S\~{a}o Paulo-SP, Brazil}

\begin{abstract}
In the present work we develop a strictly Hamiltonian approach to
Thermodynamics. A thermodynamic description based on symplectic geometry
is introduced, where all thermodynamic processes can be described
within the framework of Analytic Mechanics. Our proposal is constructed
on top of a usual symplectic manifold, where phase space is even dimensional
and one has well-defined Poisson brackets. The main idea is the introduction
of an extended phase space where thermodynamic equations of state
are realized as constraints. We are then able to apply the canonical
transformation toolkit to thermodynamic problems. Throughout this
development, Dirac's theory of constrained systems is extensively
used. To illustrate the formalism, we consider paradigmatic examples,
namely, the ideal, van der Waals and Clausius gases.
\end{abstract}


\maketitle

\section{Introduction}

The use of the theoretical structure of Analytic Mechanics, with the
introduction of a principle of stationary action and a Lagrangian
description, allowed a unification of many areas of Classical Physics,
including General Relativity. Also, this structure provides a more
straightforward way to the quantum description of important physical
systems. Both in the classical and quantum treatment, Dirac's theory
of constrained systems \cite{dir} has a central role. In this sense,
Analytic Mechanics provides a unified language for much of the Physics
landscape.

The same degree of generalization, albeit in a different path, is
given by Thermodynamics. This theory ignores the internal (microscopic)
structure of the system to be described, treating it as a ``black
box''. The final states of a given system can be described by a reduced
number of variables which effectively implement the interaction of
the black box with its environment. In this way, the description provided
by Thermodynamics is generally adequate, even if the system of interest
is a gas in a recipient or a black hole (which is fundamentally a
true black box, at least in classical terms).

Previous attempts at unifying the theoretical framework of Analytic
Mechanics and Thermodynamics followed two distinct paths, one alongside
a more physical standpoint, and another towards a geometrical unification.
As an example of the first, one can cite \cite{Peter79}, where Poisson
brackets (PB) between thermodynamic variables are defined, but the
issue of the invariance of the new PB with respect to canonical transformations
is not completely clarified. The second approach has its origin with
Gibbs \cite{gibbs} and Caratheodory \cite{caratheodory}, and culminated
with the work of Hermann \cite{hermann}. In a nutshell, the geometric
approach assigns a contact structure to the thermodynamic phase space,
such that the Legendre submanifolds describe equilibrium states. One
then defines a pseudo-Riemannian metric on the phase space which is
compatible with the contact structure. The contact structure is responsible
for encoding the first law, while the metric structure encodes the
second law \cite{mrugala1990}. Notwithstanding the conceptual clarity
of the geometric approach, the meaning of the of length of curves
in the thermodynamic state space is not completely understood given
the various proposals and interpretations \cite{salamon1983,mrugala1990,crooks2007,bravetti2015},
as well as the physical meaning of contact symmetries and contact
vector field flow, analogs of symplectomorphisms and Hamiltonian flow
in  Mechanics \cite{mrugala1993,bravetti2015}.

Another important development on the subject was presented in \cite{rajeev2008}.
In this work, a Hamilton-Jacobi formalism is proposed starting from
the geometric approach, i.e., a contact manifold. By solving the Hamilton-Jacobi
equation, one is able to recover all equations of state. In a different
work \cite{rajeev2008b}, the author quantizes the previous approach
in the framework of deformation quantization of contact manifolds.
However, the presence of non-conventional structures, such as the
odd thermodynamic phase space, Lagrange brackets and the resulting
non-standard algebra of observables, make it difficult to obtain recognizable
physical features, e.g., thermodynamic uncertainty relations.

\newpage

Contrasting with previous developments, we present a 
strictly Hamiltonian approach to Thermodynamics. Our proposal
sets aside the contact manifold framework and is constructed on top
of a usual symplectic manifold, where phase space is even dimensional
and one has well-defined Poisson brackets. The main idea is the introduction
of an extended phase space where thermodynamic equations of state
are realized as constraints. We are then able to apply the canonical
transformation toolkit to thermodynamic problems, and with little
effort we solve van der Waals and Clausius gases from the much simpler
ideal gas. Finally, our approach allows a Lagrangian description of
Thermodynamics.

The structure of this work is as follows. In section~\ref{ICPB}
we consider the symplectic structure involving thermodynamic variables
in the present approach, introducing suitable Poisson brackets and
analyzing related integrability issues. In section~\ref{eps} the
formal structure for the extended phase space is rigorously developed.
The constraint structure of our formalism and Lagrangian description
of thermodynamic systems in the context introduced here are discussed
in section~\ref{csld}. To illustrate the formalism, we consider
paradigmatic examples in section~\ref{vdw}, namely ideal, van der
Waals and Clausius gases. Final considerations and some perspectives
of future developments are presented in section~\ref{conclusion}.

\section{\label{ICPB}Integrability and Poisson brackets}

There have been works dedicated to the attempt of establishing a symplectic
structure involving thermodynamic variables \cite{Peter79,gambar1994}.
In fact, the duality between Mechanics and Thermodynamics is due to
the possibility of writing the integrability conditions for thermodynamic
variables as Poisson brackets once one has identified the thermodynamic
variables with coordinates and momenta of some phase-space.

However the analogy cannot be taken too far, since it is not a priori
clear how to translate Hamiltonian trajectories in phase-space, i.e.,
solutions of the Hamilton equations of motion $\dot{q}=\left\{ q,H\right\} $,
$\dot{p}=\left\{ p,H\right\} $, to the thermodynamic context. In
particular, a relevant question is how to interpret the evolution
parameter present in the Hamilton equations. On the other hand, the
different integrability conditions in Thermodynamics are dependent
on the chosen potential, and are related by Legendre transformations.
Another important issue is then how to translate the thermodynamic
developments back to the mechanics context. We shall address these
questions in the present work by providing a framework for incorporating
potential independence in Thermodynamics and at the same time giving
meaning to the evolution parameter for the Hamilton trajectories.

Here and in what follows, a Hamiltonian system is composed of the
triple $\left(M,\omega,X_{H}\right)$, $M$ is a smooth manifold,
$\omega$ the canonical symplectic form on $M$ and $X_{H}$ the Hamiltonian
vector field. Let us remind our reader that given a set of local Darboux
coordinates $\left(q^{i},p_{i}\right)$, $i=1,...,n$, on an open
set of $M$, with $\omega=dp_{i}\wedge dq^{i}$ , and a canonical
transformation $C:\mathbb{R}^{2n}\rightarrow\mathbb{R}^{2n}$, $C(p,q)=\left(P,Q\right)$,
it is always possible to find among the $2n$ variables $P_{i}$ and
$Q^{i}$, a set of $n$ independent variables $\left\{ y^{i}\right\} _{i=1}^{n}$
such that $\det\frac{\partial^{2}S}{\partial y^{i}\partial q^{j}}\neq0$
where $S\left(q,y\right)$ is a generating function of the canonical
transformation $C$ (see \cite{arn}). We are adopting the sum convention
(from $1$ to $n$) over repeated indices.

To fix ideas, let $q^{i}$ and $Q^{i}$ be a set of independent coordinates,
and $S\left(q,Q\right)$ the generating function satisfying $dS=p_{i}dq^{i}-P_{i}dQ^{i}$.
Then 
\begin{align}
p_{i} & =p_{i}\left(q,Q\right)=\frac{\partial S}{\partial q^{i}}\left(q,Q\right)\,,\nonumber \\
P_{i} & =P_{i}\left(q,Q\right)=-\frac{\partial S}{\partial Q^{i}}\left(q,Q\right)\,.
\end{align}
Therefore, taking into account that the $p_{i}$ are functions of
$q^{i}$ in this setting, the Poisson brackets 
\begin{equation}
\left\{ p_{i},p_{j}\right\} _{p,q}=\frac{\partial p_{i}}{\partial q^{j}}-\frac{\partial p_{j}}{\partial q^{i}}\equiv0
\end{equation}
express a subset of the integrability conditions for $S$ arising
from $d^{2}S\equiv0$, where the subscript in $\{\cdot,\cdot\}_{p,q}$
means the Poisson brackets are evaluated in the variables $q^{i}$
and $p_{i}$. By performing Legendre transformations, one can obtain
more integrability conditions. For instance, $S^{\prime}=S-p_{i}q^{i}$
is a generating function for $C$ such that $dS^{\prime}=-q^{i}dp_{i}-P_{i}dQ^{i}$,
$q^{i}=q^{i}\left(p,P\right)=-\partial S^{\prime}/\partial p_{i}$,
$P_{i}=P_{i}\left(p,P\right)=\partial S^{\prime}/\partial Q^{i}$.
It follows from $d^{2}S^{\prime}\equiv0$ that 
\begin{equation}
\left\{ q^{i},q^{j}\right\} _{p,q}=\frac{\partial q^{j}}{\partial p_{i}}-\frac{\partial q^{i}}{\partial p_{j}}\equiv0\,.
\end{equation}

As an example, let $(q^{i},p_{i})_{i=1,2}$ be canonical coordinates
in $\mathbb{R}^{4}$, and $C:\mathbb{R}^{4}\rightarrow\mathbb{R}^{4}$
a canonical transformation $(q^{i},p_{i})\mapsto(q^{\prime i},p_{i}^{\prime})$
generated by $F(q^{1},p_{2},q^{\prime1},q^{\prime2})=S(q,q^{\prime})-q^{2}p_{2}$,
where $S(q,q^{\prime})$ is as above. One has in particular $p_{1}=p_{1}(q^{1},p_{2},q^{\prime})$
and $q^{2}=q^{2}(q^{1},p_{2},q^{\prime})$. Then a subset of the integrability
conditions of $F$ can be written as 
\begin{equation}
\{p_{1},q^{2}\}_{p,q}=-\frac{\partial q^{2}}{\partial q^{1}}-\frac{\partial p_{1}}{\partial p_{2}}\equiv0~.\label{eq:integra-cond}
\end{equation}
For a thermodynamic analog with coordinates $(S,T,P,V)$ and internal
energy $dU=TdS-PdV$, we identify thermodynamic and mechanic coordinates
by $q^{1}=S$, $p_{1}=T$, $q^{2}=V$, $p_{2}=-P$. As a result, the
integrability condition~(\ref{eq:integra-cond}) gives the Maxwell
relation 
\begin{equation}
\left.\frac{\partial V}{\partial S}\right\vert _{P}=\left.\frac{\partial T}{\partial P}\right\vert _{S}\,.
\end{equation}
Considering other generating functions related to the possible different
choices of independent variables, one can obtain the remaining Maxwell
relations in a similar fashion.%
\footnote{In a context restricted to thermodynamic variables, this technique
was presented in \cite{Peter79} with the introduction of a Poisson
bracket, which in our exposition follows from the assumption of an
underlying symplectic structure.%
}

Consider now a a mechanical problem in the usual phase space $(T^{\ast}\mathbb{R},dp\wedge dq)$
with coordinates $\left(q,p\right)$. Let us extend this space by
including the new canonical pair $\left(\tau,\pi\right)$. We then
define in this extended space the Poisson bracket 
\begin{equation}
\left\{ f,g\right\} =\frac{\partial f}{\partial q}\frac{\partial g}{\partial p}-\frac{\partial f}{\partial p}\frac{\partial g}{\partial q}-\left(\frac{\partial f}{\partial\tau}\frac{\partial g}{\partial\pi}-\frac{\partial f}{\partial\pi}\frac{\partial g}{\partial\tau}\right)\,,
\end{equation}
which gives the canonical relations $\left\{ q,p\right\} =\left\{ \pi,\tau\right\} =1$
and $\left\{ p,\pi\right\} =\left\{ q,\pi\right\} =0$. We can choose
any non-canonical pair as independent variables. For example, by choosing
$\left(q,\tau\right)$ as independent variables, we can evaluate the
PB of $p=p\left(q,\tau\right)$ and $\pi=\pi\left(q,\tau\right)$,
\begin{equation}
\left\{ p,\pi\right\} =\frac{\partial p}{\partial q}\frac{\partial\pi}{\partial p}-\frac{\partial p}{\partial p}\frac{\partial\pi}{\partial q}-\frac{\partial p}{\partial\tau}=\left\{ p,\pi\right\} _{p,q}-\frac{\partial p}{\partial\tau}\,,
\end{equation}
where $\left\{ p,\pi\right\} _{p,q}$ is the usual Poisson bracket
in the phase space $(T^{\ast}\mathbb{R},dp\wedge dq)$. From the canonical
relation $\left\{ p,\pi\right\} =0$ we have $\left\{ p,\pi\right\} _{p,q}=\partial p/\partial\tau$.
If $p$ has no explicit dependence in $q$ and we can identify $\pi$
with the Hamiltonian of the system, this expression becomes one of
Hamilton's equation.

So, if we can actually enforce these conditions, that is, if we can
construct a mechanical description in an ``extended space\textquotedblright{},
we will be able to treat mechanic and thermodynamic problems on the
same footing. In the next section we will demonstrate how this construction
can be implemented using Dirac theory for systems with constraints.

\section{\label{eps}Extended phase space}

We now detail the formal structure for the extended phase space and
see how the original Hamiltonian system can be recast in this space.
Let $dp_{i}\wedge dq^{i}$ be the symplectic form in local coordinates
in an open set of the cotangent bundle $T^{\ast}Q$, where $Q$ is
some configuration space manifold of dimension $n$. Let $N=Q\times\mathbb{R}$
and consider now the cotangent bundle $T^{\ast}N$ with a symplectic
form given by 
\begin{equation}
\omega=dp_{i}\wedge dq^{i}+d\pi\wedge d\tau\,.
\end{equation}
Let $H:T^{\ast}(N)\rightarrow\mathbb{R}$ be a function whose expression
in local coordinates is given by 
\begin{equation}
H=\pi+h(q,p,\tau)\,,
\end{equation}
where $h(q,p,\tau)$ is some function on the phase space $T^{\ast}N$.
In coordinates, the Poisson brackets have the expression 
\begin{equation}
\{f,g\}=\frac{\partial f}{\partial q^{i}}\frac{\partial g}{\partial p_{i}}-\frac{\partial f}{\partial p_{i}}\frac{\partial g}{\partial q^{i}}+\frac{\partial f}{\partial\tau}\frac{\partial g}{\partial\pi}-\frac{\partial f}{\partial\pi}\frac{\partial g}{\partial\tau}\,,\label{eq:poisson}
\end{equation}
where $f$ and $g$ are functions on $T^{\ast}N$. Hence, the nonvanishing
canonical PB are 
\begin{equation}
\left\{ q,p\right\} =\left\{ \tau,\pi\right\} =1\,.
\end{equation}
Let $(p(t),q(t))$ be a trajectory in the phase space $T^{\ast}N$.
The Hamiltonian phase flux is generated by the field 
\begin{equation}
X_{H}=\frac{\partial}{\partial\tau}+\frac{\partial h}{\partial p_{i}}\frac{\partial}{\partial q^{i}}-\frac{\partial h}{\partial q^{i}}\frac{\partial}{\partial p_{i}}-\frac{\partial h}{\partial\tau}\frac{\partial}{\partial\pi}\,.
\end{equation}
Thus, an integral curve $\gamma$ of $X_{H}$, i.e., a curve that
satisfies $\dot{\gamma}=X_{H}(\gamma(t))$, gives the Hamilton equations
\begin{align}
 & \frac{dq^{i}}{dt}=X_{H}(q)=\frac{\partial h}{\partial p_{i}}~,\ \frac{dp_{i}}{dt}=X_{H}(p)=-\frac{\partial h}{\partial q^{i}}\,,\nonumber \\
 & \frac{d\pi}{dt}=X_{H}(\pi)=-\frac{\partial h}{\partial\tau}~,\ \frac{d\tau}{dt}=X_{H}(\tau)=1\,.
\end{align}
If $\eta$ denotes the set of coordinates in $T^{\ast}N$, one has
\begin{equation}
\frac{d\eta}{dt}=X_{H}(\eta)=\{\eta,H\}\,.
\end{equation}
For $f$ a function that does not depend on $\pi$, its time evolution
is given by 
\begin{equation}
\frac{df}{dt}=X_{H}(f)=\{f,H\}=\frac{\partial f}{\partial\tau}+\frac{\partial f}{\partial q^{i}}\frac{\partial h}{\partial p_{i}}-\frac{\partial f}{\partial p_{i}}\frac{\partial h}{\partial q^{i}}\,.
\end{equation}

If one considers the constraint surface $H=0$ and the relation $d\tau=dt$
resulting from the time evolution equation of $\tau$, then the tautological
form in $T^{\ast}N$ degenerates to the Poincar\'{e}-Cartan form in $T^{\ast}Q\times\mathbb{R}$
with Hamiltonian $h$: 
\begin{equation}
\left.p_{i}dq^{i}+\pi d\tau\right\vert _{H=0}=p_{i}dq^{i}-hdt\,.
\end{equation}
Therefore, we can formulate the mechanics on the contact space $T^{\ast}Q\times\mathbb{R}$
as mechanics on the extended phase space $T^{\ast}N$ with the constraint
$H=0$, whose canonical Hamilton function is $H_{c}=\lambda H$ and
$\lambda$ is a Lagrange multiplier, plus possible additional linear
combinations of primary constraints \cite{dir,gitman-tyutin}.

We will now impose the chronological constraint 
\begin{equation}
\psi=\tau-t\,.\label{vc}
\end{equation}
This constraint formalizes the idea of time as a phase space coordinate.
With this, we obtain a second class constraint theory, since $\{\psi,H\}=1$.
The conservation in time of the constraint $\psi$ is 
\begin{equation}
\frac{d\psi}{dt}=\frac{\partial\psi}{\partial t}+\{\psi,\lambda H\}=-1+\lambda=0\,.
\end{equation}
Hence, we have that $\lambda=1$, and the canonical Hamiltonian is
simply $H$. We can also make the canonical transformation to new
variables $q^{\prime i}=q^{i}$, $p_{i}^{\prime}=p_{i}$, $\tau^{\prime}=\tau-t$
and $\pi^{\prime}=\pi$, given by the generating function $W(q,\tau,p^{\prime},\pi^{\prime},t)=q^{i}p_{i}^{\prime}+\left(\tau-t\right)\pi^{\prime}$.
It follows that the constraint surface becomes 
\begin{equation}
\phi_{1}=\pi+h(q,p,\tau^{\prime}+t)=0\,,\,\phi_{2}=\tau^{\prime}=0\,,
\end{equation}
and the Hamiltonian is $H^{\prime}=H+\partial W/\partial t=h$. The
time evolution of the quantities $\eta=\left(q,p,\tau^{\prime},\pi\right)$
is given in terms of Dirac brackets \cite{gitman-tyutin}, 
\begin{equation}
\frac{d\eta}{dt}=\{\eta,h\}_{D}\,.
\end{equation}
However, due to constraint relations between constraints and the reduced
 Hamiltonian $h$, the Dirac brackets reduce to Poisson brackets 
\begin{align}
 & \frac{dq^{i}}{dt}=\{q^{i},h\}=\frac{\partial h}{\partial p_{i}}\,,\,\frac{dp_{i}}{dt}=\{p_{i},h\}=-\frac{\partial h}{\partial q^{i}}\,,\,\\
 & \frac{d\pi}{dt}=0\,\,,\frac{d\tau^{\prime}}{dt}=0\,,\phi_{1}=\phi_{2}=0\,,
\end{align}
and therefore 
\begin{equation}
\frac{d}{dt}f(q,p,t)=\left.\{f,h\}\right\vert _{\phi=0}=\frac{\partial f}{\partial t}+\frac{\partial f}{\partial q^{i}}\frac{\partial h}{\partial p_{i}}-\frac{\partial f}{\partial p_{i}}\frac{\partial h}{\partial q^{i}}\,.
\end{equation}
In this way, the initial dynamics in the extended phase space with
Hamiltonian $\lambda H$ and gauge fixing condition $\tau=t$ is equivalent
to the dynamics in phase space $(q,p)$ with the reduced Hamiltonian
$h$.

\section{\label{csld}Constraint structure and Lagrangian}

Let $\left\{ q^{i}\right\} _{i=1}^{n}$ denote the set of extensive
parameters of a thermodynamic system, such as volume or entropy. Then
the internal energy $u$ is a first-order homogeneous function of
$q^{i}$, $u=u\left(q^{1},...,q^{n}\right)$. In order that the thermodynamic
system be completely specified, one needs $n$ equations of state,
of the form 
\begin{equation}
p_{i}=\frac{\partial u}{\partial q^{i}}\left(q^{1},...,q^{n}\right)\doteq f_{i}\left(q^{1},\ldots,q^{n}\right)\,,
\end{equation}
where $p_{i}$ are intensive parameters, such as temperature or pressure.
The above relations give rise to $n$ constraints in the Hamiltonian
formalism, where $q^{i}$ are coordinates and $p_{i}$ are conjugate
momenta. These correspond precisely to a system with $n$ degrees
of freedom and $n$ primary constraints $\Phi_{i}=p_{i}-f_{i}(q)$.

Furthermore, because given two states in the thermodynamic configuration
space any trajectory connecting them must be a valid thermodynamic
path, there are no physical degrees of freedom in the corresponding
mechanical analog. As a result, either the number of first class constraints
is $n$, or this number is $k$, in case there are $p$ second-class
constraints $\chi_{i}$ among the constraints, such that $n=k+p/2$.
In either case, we are able to show that the Lagrange function 
\begin{equation}
L(q,\dot{q})=\left.p\dot{q}-H_{c}\right\vert _{p=p(q,\dot{q})}\label{lg}
\end{equation}
is a total derivative, where the Hamiltonian $H_{c}$ is a linear
combination of primary first-class constraints $\Phi_{i}$. 

\begin{proposition} \label{proposition1} Let the total set of irreducible
constraints $\{\Phi_{i}\}_{i=1}^{n}$ of a Hamiltonian system be time-independent
primary first-class constraints. Then the Lagrange function is a total
derivative. \end{proposition}

\begin{proof} The constraints have the structure $\Phi_{i}=p_{i}-f_{i}(q)$,
because they are primary, and the fact that they Poisson-commute gives
the integrability condition 
\begin{equation}
\frac{\partial f_{i}}{\partial q^{j}}-\frac{\partial f_{j}}{\partial q^{i}}=0\ ,
\end{equation}
which implies that $f_{i}$ are the components of a gradient, $f_{i}=\partial\phi/\partial q^{i}$.
Since there are $n$ first-class constraints, there are no degrees
of freedom, so the Hamiltonian $H_{c}$ is proportional to constraints
\begin{equation}
H_{c}=\lambda^{i}\Phi_{i}\ ,
\end{equation}
where the Lagrange multipliers $\lambda^{i}$ are undetermined velocities,
$\lambda^{i}=\dot{q}^{i}$. Therefore, the Lagrange function is 
\begin{equation}
L=p_{i}\dot{q}^{i}-H_{c}=\dot{q}^{i}\frac{\partial\phi}{\partial q^{i}}=\frac{d\phi}{dt}~.
\end{equation}

\end{proof}

Proposition~\ref{proposition1} can be generalized to a system with
second-class constraints, provided the total number of degrees of
freedom stays the same. This is done in proposition~\ref{proposition2},
presented in the following.

\begin{proposition} \label{proposition2} Let $\{\Phi_{i}\}_{i=1}^{k}$
be a set of irreducible primary time-independent first-class constraints,
and $\{\chi_{i}\}_{i=1}^{p}$ a set of second-class constraints, such
that $n=k+p/2$, where $2n$ is the dimension of the symplectic manifold.
Then the Lagrange function is a total derivative. \end{proposition}

\begin{proof} The restriction $n=k+p/2$ implies there are no degrees
of freedom, so the Hamiltonian $H_{c}$ is as before proportional
to constraints 
\begin{equation}
H_{c}=\sum_{i=1}^{k}\lambda^{i}\Phi_{i}+\sum_{i=1}^{p}\beta^{i}\chi_{i}\,.
\end{equation}
The condition of preservation of the second-class constraints $\chi_{i}$
in time gives $\beta^{i}\equiv0$. Therefore, taking into account
the vanishing of the corresponding velocities $\beta^{i}=\dot{q}^{i}$
in the Lagrangian, one has the same situation as in Proposition~\ref{proposition1}.
Therefore, the Lagrange function is a total derivative. \end{proof}

Any theory constructed in the extended phase space can be dealt with
this formalism. Thus, we are able to propose a Lagrangian description
for the thermodynamic system. Moreover, the fact that the total Hamiltonian
vanishes on the constraint surface implies the Lagrangian is first-order
homogeneous in the velocities. In the example provided below, the
Lagrangian turns out to be the time-derivative of the internal energy
of the thermodynamic analog. In addition, since all mechanical systems
have vanishing physical degrees of freedom, one can always find a
canonical transformation connecting two systems of the same dimension.

\section{\label{vdw}Applications}

To illustrate the formalism and highlight its main advantages, in
what follows we apply the results presented so far to construct mechanical
analogs for the ideal, van der Waals and Clausius gases.

\subsection{Mechanical setup for ideal and van der Waals gases}

Both ideal and van der Waals gases can be described by the pairs of
conjugate variables%
\footnote{The lower case letters represent specific quantities, e.g., $s=S/N$,
where $N$ is the number of particles.%
} $\left(s,T\right)$ and $\left(v,-P\right)$, such that the internal
energy $u$ satisfies the relation $du=Tds-Pdv$. For instance, in
the energy representation, the van der Waals gas is described by the
equations of state 
\begin{equation}
T\left(u,v\right)=\frac{2}{3}\left(u+\frac{a}{v}\right)\,,\,\, P\left(T,v\right)=\frac{T}{v-b}-\frac{a}{v^{2}}\,,\label{vdw-eq-state}
\end{equation}
with non zero constants $a$ and $b$. The case $a=b=0$ in Eq.~(\ref{vdw-eq-state})
represents the ideal gas.

The symplectic space for the construction is the phase space $\mathbb{R}^{4}$
with coordinates $\left(q,p,\tau,\pi\right)$, such that the Poisson
brackets are given by 
\begin{equation}
\left\{ f,g\right\} =\frac{\partial f}{\partial q}\frac{\partial g}{\partial p}+\frac{\partial f}{\partial\tau}\frac{\partial g}{\partial\pi}-\frac{\partial g}{\partial q}\frac{\partial f}{\partial p}-\frac{\partial g}{\partial\tau}\frac{\partial f}{\partial\pi}\,.
\end{equation}
In this space we introduce the Hamilton function $H=\pi+h\left(q,p,\tau\right)$,
where $h\left(q,p,\tau\right)$ is the reduced Hamiltonian, to be
determined. We have shown that on the surface $H=0$ and $\tau-t=0$
the given Hamiltonian system simplifies to the system with reduced
Hamilton function $h$ in the reduced phase space $\left(q,p\right)$.

In order to map the thermodynamic variables to mechanic ones, let
us identify the tautological form $\theta$ with $du$, $\theta=pdq+\pi d\tau\equiv du$.
A possible dictionary between variables is 
\begin{equation}
\tau=s,\,\pi=T\,,q=v\,,p=-P\ .\label{eq:mec-thermo-dic}
\end{equation}
Other choices, as we will show, differ by a canonical transformation.
In this way, equations of state are translated as 
\begin{equation}
\pi=\left(q-b\right)\left(\frac{a}{q^{2}}-p\right)\,,\,\,\pi=\frac{2}{3}\left(u+\frac{a}{q}\right)\,.\label{vdwha}
\end{equation}

\subsection{Solving the ideal gas within a mechanical framework}

We proceed setting $a=b=0$, i.e, we will work with the ideal gas
applying the formalism introduced, in order to obtain a fundamental
equation. After solving this problem we show how to obtain the solution
for the van der Waals equation using a canonical transformation. By
doing this we can not only apply the formalism, but also illustrate
how the canonical transformation technique can be used to solve a
thermodynamic problem. 

For the ideal gas, Eq.~(\ref{vdwha}) takes the simpler form 
\begin{equation}
\pi=-qp\,,\,\,\pi=2u/3\,.\label{ideal}
\end{equation}
In order to incorporate the identification $\theta\equiv du$, we
take the exterior derivative of the second relation in Eq.~(\ref{ideal}),
\begin{equation}
d\pi=\frac{2}{3}\left(pdq+\pi d\tau\right)\,.\label{eq:diff}
\end{equation}
These relations provide two constraints in phase space. Substituting
in Eq.~(\ref{eq:diff}) the first relation in Eq.~(\ref{ideal}),
we have 
\begin{equation}
d\pi=\frac{2}{3}\left(pdq-qpd\tau\right)=-d\left(pq\right)\,.
\end{equation}
Integrating, one obtains 
\begin{equation}
p=-Ae^{\frac{2}{3}\tau}q^{-\frac{5}{3}}\,,
\end{equation}
where $A$ is a constant. Therefore, we get the constraint 
\begin{equation}
\phi\left(q,p,\tau\right)=p+Ae^{\frac{2}{3}\tau}q^{-\frac{5}{3}}\,.\label{eq:constraint-phi}
\end{equation}
The first relation in Eq.~(\ref{ideal}) states that $\pi+qp=0$.
We use the first constraint (\ref{eq:constraint-phi}) to eliminate
the momentum $p$ and write the second constraint in the form of a
primary constraint, 
\begin{equation}
H=\pi-Ae^{\frac{2}{3}\tau}q^{-\frac{2}{3}}\,.
\end{equation}
That is, the reduced Hamilton function is 
\begin{equation}
h\left(q,\tau\right)=-Ae^{\frac{2}{3}\tau}q^{-\frac{2}{3}}\,.
\end{equation}
As a result, the canonical Hamilton function is given by $H_{c}=\sigma H+\lambda\phi$,
where $\sigma$ and $\lambda$ are Lagrange multipliers.

The conservation equations for the constraints are proportional to
constraints, so they do not fix the Lagrange multipliers, 
\begin{align}
 & \frac{dH}{dt}=\left\{ H,H_{c}\right\} =\lambda\left\{ H,\phi\right\} =\lambda\phi\,,\\
 & \frac{d\phi}{dt}=\left\{ \phi,H_{c}\right\} =\sigma\left\{ \phi,H\right\} =-\sigma\phi\,.
\end{align}
The constraints are thus preserved on the constraint surface. The
total set of constraints is first-class. Therefore, as expected, there
are no degrees of freedom. In what follows, we shall impose the chronological
gauge (\ref{vc}) $\psi=\tau-t$. As a result, one has $\sigma=1$
and $H_{c}=H+\lambda\phi$, where $\lambda\left(t\right)$ is an arbitrary
function of time which embodies the gauge freedom of the model.

Going back to thermodynamic variables (\ref{eq:mec-thermo-dic}),
we see that the constraints $H=0$ and $\phi=0$ give the equations
\begin{equation}
T-vP=0\,,\,\, P-Ae^{\frac{2}{3}s}v^{-\frac{5}{3}}=0\,.
\end{equation}
These constraints, together with the second relation in Eq.~(\ref{ideal}),
give us 
\begin{equation}
u\left(s,v\right)=\frac{3}{2}\frac{A}{v^{2/3}}\exp\left(\frac{2}{3}s\right)\,,\label{ideal-internal-energy}
\end{equation}
which is the fundamental equation of the ideal gas in the entropy
representation. The constant $A$ expresses the freedom in the zero
value of the entropy (that can be fixed by Nernst theorem).

\subsection{From the ideal gas to the van der Waals gas}

As mentioned in the previous section, the above development also allows
us to write a Lagrangian description for the ideal gas. The Lagrange
function for this model is given by Eq.~(\ref{lg}), 
\begin{equation}
L\left(q,\dot{q},\tau\right)=Ae^{\frac{2}{3}\tau}q^{-\frac{5}{3}}\left(\dot{\tau}q-\dot{q}\right)\,.\label{lvd}
\end{equation}
Not surprisingly, the Lagrangian is first-order homogeneous in the
velocities, and a total derivative, $L\left(q,\dot{q},\tau\right)=du/dt$,
where $u$ is the internal energy (\ref{ideal-internal-energy}).
In other words, the Lagrange function obtained and Dirac's theory
for constrained systems together give the fundamental equation for
the ideal gas.

We can now show explicitly a canonical transformation connecting the
ideal gas to the van der Waals gas, as an example of the fact that
there must be such a transformation connecting two theories in $\mathbb{R}^{4}$
with vanishing physical degrees of freedom. Let us use primes to indicate
the quantities referring to the van der Waals gas ($u^{\prime},\pi^{\prime},q^{\prime},p^{\prime},\tau^{\prime}$).
Comparing the first expression in Eq.~(\ref{ideal}) with the first
expression in Eq.~(\ref{vdwha}), we see that the corresponding constraints
can be related by means of the canonical transformation 
\begin{equation}
q=q^{\prime}-b~,\ p=p^{\prime}-aq^{\prime-2}\;,\ \pi=\pi^{\prime}~,\;\tau=\tau^{\prime}.\label{t4}
\end{equation}
Thus, the total set of constraints for the van der Waals gas is 
\begin{align}
H^{\prime}\left(\eta^{\prime}\right) & =H\left(\eta\left(\eta^{\prime}\right)\right)=\pi^{\prime}-Ae^{\frac{2}{3}\tau^{\prime}}(q^{\prime}-b)^{-\frac{2}{3}}~,\nonumber \\
\phi^{\prime}\left(\eta^{\prime}\right) & =\phi\left(\eta\left(\eta^{\prime}\right)\right)=p^{\prime}-\frac{a}{q^{\prime2}}+\frac{Ae^{\frac{2}{3}\tau^{\prime}}}{\left(q^{\prime}-b\right)^{\frac{5}{3}}}\;.
\end{align}
where $\eta=\left(q,\tau,\pi\right)$. Since the transformation is
time-independent, the transformed Hamilton function is 
\begin{equation}
H_{c}^{\prime}=\sigma H^{\prime}+\lambda\phi^{\prime}~.
\end{equation}
Inserting the canonical transformations in the tautological form $du=pdq+\pi d\tau$,
one has 
\begin{equation}
du=\left(p^{\prime}-aq^{\prime-2}\right)dq^{\prime}+\pi^{\prime}d\tau^{\prime} 
\label{t1}
\end{equation}
which gives the second equation in (\ref{vdw-eq-state}). Using the
fact that the transformation is canonical, the transformed tautological
form $du^{\prime}=p^{\prime}dq^{\prime}+\pi^{\prime}d\tau^{\prime}$
differs from $du$ by an exact differential, 
\begin{equation}
du-du^{\prime}=ad\left(q^{\prime-1}\right)\Rightarrow u=u^{\prime}+\frac{a}{q^{\prime}}~.
\end{equation}
Finally, from the expression for $u$ in Eq.~(\ref{ideal-internal-energy}),
we get the van der Waals equation of state 
\begin{equation}
u^{\prime}=\frac{3}{2}\frac{Ae^{\frac{2}{3}\tau^{\prime}}}{\left(q^{\prime}-b\right)^{2/3}}-\frac{a}{q^{\prime}}\doteq\frac{3}{2}\frac{A}{\left(v-b\right)^{2/3}}\exp\left(\frac{2}{3}s\right)-\frac{a}{v}~.\label{eq:vdw-internal-energy-1}
\end{equation}

\subsection{Canonical transformations and gauge fixing}

If, instead of using the thermodynamics-mechanics dictionary in Eq.~(\ref{eq:mec-thermo-dic}), we had initially made the identification
\begin{equation}
p=T~,\ q=s~,\ \pi=-P~,\ \tau=v~,\label{eq:mec-thermo-dic2}
\end{equation}
that is, if we had made the canonical transformation $\left(q,p,\tau,\pi\right)\mapsto\left(q^{\prime},p^{\prime},\tau^{\prime},\pi^{\prime}\right)$,
where $q^{\prime}=\tau=s$, $p^{\prime}=\pi=T$, $\tau^{\prime}=q=v$,
$\pi^{\prime}=p=-P$, the canonical Hamilton function would
be $H_{c}=H^{\prime}+\lambda\phi^{\prime}$, where, for the ideal
gas, 
\begin{align}
 & H^{\prime}(\eta^{\prime})=H\left(\eta\left(\eta^{\prime}\right)\right)=p^{\prime}+\tau^{\prime}\pi^{\prime}\,,\label{id1}\\
 & \phi^{\prime}(\eta^{\prime})=\pi^{\prime}+Ae^{\frac{2}{3}q^{\prime}}\tau^{\prime-\frac{5}{3}}\,.\label{id2}
\end{align}
One can locally transform $\left(p,\pi,q,\tau\right)\mapsto\left(p^{\prime},\pi^{\prime},q^{\prime},\tau^{\prime}\right)$
by means of the generating function of the second kind $W\left(q,\tau,p^{\prime},\pi^{\prime}\right)=\pi^{\prime}q+p^{\prime}\tau$.
Then, as $d(u-u^{\prime})=dW$, one has 
\begin{equation}
u^{\prime}=u-Ts+Pv~,
\end{equation}
which is the Gibbs free energy.

Let us now see how the gauge fixing of $\lambda$ manifests itself
in the thermodynamic description. Let us consider the isobaric process
$p=p_{0}$. This constraint, which we denote $\tilde{\phi}=p-p_{0}$,
allows us to fix $\lambda$: 
\begin{equation}
\left\{ \tilde{\phi},H_{c}\right\} _{\Phi=0}=0\Leftrightarrow\lambda=-q-\frac{2\pi}{p}\,.
\end{equation}
The evolution of $p$ gives trivially 
\begin{equation}
\frac{dp}{dt}=\left\{ p,H_{c}\right\} _{\Phi=0}=0\,,
\end{equation}
that is, the pressure $p$ is constant. On the other hand, 
\begin{equation}
\frac{dq}{dt}=\left\{ q,H_{c}\right\} _{\Phi=0}=-\frac{2\pi}{5p_{0}}\,.\label{eq:dqdt}
\end{equation}
The temperature $\pi$ has to naturally compensate the pressure, so
the process remains isobaric: 
\begin{equation}
\frac{d\pi}{dt}=\left\{ \pi,H_{c}\right\} _{\Phi=0}=-\frac{2}{3}p_{0}\left(q+\frac{2\pi}{5p_{0}}\right)\,.\label{dpi_dq}
\end{equation}
Indeed, from Eq.~(\ref{eq:dqdt}) and Eq.~(\ref{dpi_dq}) one has
\begin{equation}
\frac{d\pi}{dq}=-p_{0}\,,
\end{equation}
which is precisely the temperature variation with relation to volume
$\partial\pi/\partial q$, that we obtain from the equation of state
\begin{equation}
\pi\left(q\right)=-p_{0}q\,.
\end{equation}
We thus see the imposition of supplementary conditions to the equations of state is akin to gauge fixing the theory.

\subsection{Clausius gas}

As an additional example of application of the method introduced in
the present work, let us consider the Clausius gas. This system is
described by the following equation of state, 
\begin{equation}
P=\frac{T}{v-b}+\frac{a}{T\left(v-c\right)^{2}}~.\label{cl}
\end{equation}
Since it is not possible to isolate $T$ in Eq.~(\ref{cl}), the
identification proposed in Eq.~(\ref{eq:mec-thermo-dic}) is not
appropriate. However, with the dictionary presented in Eq.~(\ref{eq:mec-thermo-dic2}),
the mechanical analog of Eq.~(\ref{cl}) is 
\begin{equation}
\tilde{H}=\tilde{p}+\left(\tilde{\tau}-b\right)\left[\tilde{\pi}+\frac{a}{\tilde{p}\left(\tilde{\tau}-c\right)^{2}}\right]=0~,\label{pi}
\end{equation}
where in Eq.~(\ref{pi}) we are using ``tildes''
to indicate Clausius gas quantities.

We can transform expression~(\ref{pi}) into the constraint 
\begin{equation}
H^{\prime}\left(\eta^{\prime}\right)=\tilde{H}\left(\tilde{\eta}\left(\eta^{\prime}\right)\right)=p^{\prime}+\tau^{\prime}\pi^{\prime}\,,
\end{equation}
associated to the ideal gas (\ref{id1}) in coordinates given in Eq.~(\ref{eq:mec-thermo-dic2}).
The canonical transformation that relates the two sets of coordinates
is given by 
\begin{equation}
\tilde{\tau}=\tau^{\prime}+b\,,\,\,\tilde{\pi}=\pi^{\prime}-\frac{ap^{\prime-1}}{\left(\tau^{\prime}+b-c\right)^{2}}\,,\,\,\tilde{q}=q^{\prime}+\frac{ap^{\prime-2}}{\tau^{\prime}+b-c}\,,\,\,\tilde{p}=p^{\prime}\,.\label{trans}
\end{equation}
It follows that the constraint in Eq.~(\ref{id2}) becomes 
\begin{equation}
\tilde{\phi}=\phi^{\prime}\left(\eta^{\prime}\left(\tilde{\eta}\right)\right)=\tilde{\pi}+\frac{a\tilde{p}^{-1}}{\left(\tilde{\tau}-c\right)^{2}}+\frac{A}{\left(\tilde{\tau}-b\right)^{\frac{5}{3}}}\exp\left[\frac{2}{3}\left(\tilde{q}-\frac{a\tilde{p}^{-2}}{\tilde{\tau}-c}\right)\right]~.
\end{equation}

One can obtain the internal energy for the Clausius gas as was done
previously for the van der Waals gas, by calculating the difference
between the tautological forms in their tilde and prime versions:
\begin{equation}
du^{\prime}=d\tilde{u}-d\left(\frac{1}{\tilde{p}}\frac{2a}{\left(\tilde{\tau}-c\right)}\right)~.
\end{equation}
As a result, the internal energy is 
\begin{equation}
\tilde{u}=u^{\prime}+\frac{1}{\tilde{p}}\frac{2a}{\left(\tilde{\tau}-c\right)}=\frac{3}{2}T+\frac{1}{T}\frac{2a}{\left(v-c\right)}~.\label{ucl}
\end{equation}

Setting $\tilde{\phi}=0$ and using Eq.~(\ref{pi}) we have 
\begin{equation}
s=\frac{a}{T^{2}\left(v-c\right)}+\ln\left(v-b\right)+\frac{3}{2}\ln\left(\frac{T}{A}\right)~,\label{s_clausius}
\end{equation}
where we employed the thermodynamic variables presented in Eq.~(\ref{eq:mec-thermo-dic2}).
Using Eq.~(\ref{s_clausius}) and Eq.~(\ref{ucl}), one can construct
a fundamental equation for the Clausius gas. For example, the Helmholtz
free energy $f=\tilde{u}-Ts$ becomes 
\begin{equation}
f=\frac{a}{T\left(v-c\right)}+\frac{3}{2}T\left[1-\ln\frac{T}{A}-\ln\left(v-b\right)^{\frac{2}{3}}\right]~.
\end{equation}

\section{\label{conclusion}Conclusions and perspectives}

In this work we present a new formalism for constructing mechanic
analogs of thermodynamic systems. Taking previous developments into
account, we believe that in the present work we have a complete duality
between Analytic Mechanics and Thermodynamics. In fact, thermodynamic
systems (as described here) can be completely characterized in the
language of Dirac's theory of constrained systems. 

One important feature of the theoretical framework introduced is the
possibility of a Lagrangian formulation. The fact that a thermodynamic
system is not dynamical from the viewpoint of Mechanics implies that
the associated Lagrangian is a total derivative in time. This is an
important point, because as a result it is always possible to construct
a canonical transformation associating any two thermodynamic systems
with the same number of mechanical degrees of freedom. Moreover, the
primitive function for this Lagrangian furnishes a fundamental equation
for the thermodynamic description. 

We study the particular cases of the ideal, van der Waals and Clausius
gases, and we show how different dictionaries between mechanic and
thermodynamic variables are related by canonical transformations.
Moreover, these transformations also relate thermodynamic potentials.
Furthermore, we explicitly verify how the gauge freedom of the mechanical
analog is associated with the restrictions present in thermodynamic
processes. As an illustrative application, we have easily obtained
the solution of the van der Waals and Clausius gases from the (much
simpler) ideal gas solution.

It is interesting to remark that the proposed formalism gives a very
simple development for the Clausius gas system. A problem involving
a set of partial differential equations was transformed into an algebraic
problem. Of course, it is not always straightforward to find a canonical
transformation that associates a given system to the ideal gas (for
example). Still, in principle, the approach has the potential to be
extremely helpful.

It should be stressed that Dirac's theory not only gives adequate
tools for the treatment of systems with constraints, but also it provides
a guide to the quantization of these systems. The conjecture that
conjugated thermodynamic variables, e.g., pressure and volume, obey
``uncertainty relations'' was already proposed
in quantum mechanics in the early days by Bohr and Heisenberg. They
even obtained an explicit form for this uncertainty principle involving
internal energy and temperature \cite{bh}. From a more formal perspective,
using arguments based on statistical mechanics, these same relations
were again obtained more recently \cite{Rosen,Mand,Sch,WilWl2011}.
The scenario sketched here suggests that Thermodynamics has intrinsic
uncertainty relations, which could lead to noncommutativity of the
thermodynamic variables.

Still, up until now, there has been no solid framework for obtaining
thermodynamic uncertainty relations. We expect that the approach we
have developed, complemented with canonical quantization, should provide
a better understanding of the thermodynamic uncertainty relations.
Besides, since we have a usual mechanical system describing Thermodynamics,
it is clear how to develop the Hamilton-Jacobi formalism, albeit the
presence of constraints \cite{Rothe2003,Pimentel2014}. These ideas will be considered in a future
development of the present work.

\bigskip

\begin{acknowledgments}
C. M. is supported by FAPESP, Brazil (grant No.2015/24380-2) and CNPq, Brazil (grant No.307709/2015-9).
\end{acknowledgments}

\end{document}